\documentclass{llncs}
\usepackage{amsmath,amssymb,tikz}
\usetikzlibrary{snakes}
\usepackage{multicol}
\usepackage{array}
\usepackage{wrapfig}
\newtheorem{defi}{Definition}
\newtheorem{lem}[defi]{Lemma}
\newtheorem{thm}[defi]{Theorem}
\newtheorem{ex}[defi]{Example}

\def\id{\operatorname{id}}
\def\<#1>{\langle#1\rangle}
\let\set\mathbb
\def\sem{_{\mathrm{sem}}}%
\def\syn{_{\mathrm{syn}}}%
\def\BF{\operatorname{BF}} 
\def\strat{\mathbb{S}} 
\def\clap#1{\hbox to0pt{\hss#1\hss}}

\author{Manuel Kauers\inst{1} \and Martina Seidl\inst{2}}

\institute{Institute for Algebra, JKU Linz, Austria\\
\email{manuel.kauers@jku.at}
 \and
Institute for Formal Models and Verification, JKU Linz, Austria\\
\email{martina.seidl@jku.at}
}
\title{Symmetries of Quantified Boolean Formulas\thanks{Parts of this work
were supported by the Austrian Science Fund (FWF) under grant numbers
NFN S11408-N23 (RiSE),
Y464-N18, and SFB F5004.}}
\begin{document}
\maketitle

\begin{abstract}
  While symmetries are well understood for Boolean formulas
  and successfully exploited in practical SAT solving,
  less is known about symmetries in quantified Boolean formulas (QBF).
  There are some works introducing adaptions of 
  propositional symmetry breaking techniques, 
  with a theory covering only very specific parts of QBF symmetries. 
  We present a general framework that gives
  a concise characterization of symmetries of QBF. 
  Our framework naturally incorporates the duality of 
  universal and existential symmetries resulting in a general 
  basis for QBF symmetry breaking.
\end{abstract}

\section{Introduction}\label{sec:intro}

Mathematicians are generally advised~\cite{polya1945} not to destroy symmetry in a given problem but instead to exploit it.
In automated reasoning, we generally exploit symmetries by destroying them. In this context, to destroy a symmetry
means to enrich the given problem by additional constraints which tell the solver that certain parts of the search
space are equivalent, so that it investigates only one of them.
Such symmetry breaking techniques have been studied since long.
They are particularly well developed in SAT~\cite{DBLP:series/faia/Sakallah09} and CSP~\cite{DBLP:reference/fai/GentPP06}.
In CSP~\cite{DBLP:conf/aaai/CohenJJPS06} it has been observed that it is appropriate to distinguish two kinds of symmetries: those of the
problem itself and those of the solution set.
In the present paper, we apply this idea to Quantified Boolean Formulas (QBF).

\begin{wraptable}{r}{5cm}
\vspace{-.75cm}
\scriptsize
\centering

\begin{tabular}{|r|>{\raggedleft\arraybackslash}p{0.9cm}|%
>{\raggedleft\arraybackslash}p{0.9cm}|%
>{\raggedleft\arraybackslash}p{0.9cm}|%
>{\raggedleft\arraybackslash}p{0.9cm}|}
\multicolumn{1}{l}{} &
\multicolumn{4}{c}{Solving times (in sec)} \\
\multicolumn{1}{l}{} &
\multicolumn{2}{c}{w/o SB} &
\multicolumn{2}{c}{with SB}\\
\hline
$n$ &  QRes & LD & QRes & LD \\
\hline
$10$ & $0.3$ & $0.5$ & $0.4$ & $0.4$\\
$20$ & $160$ & $0.5$ & $0.4$ & $0.4$\\
$40$ & $>3600$ & $0.5$ & $0.4$ & $0.4$\\
$80$ & $>3600$ & $0.7$ & $0.4$ & $0.4$\\
$160$ &  $>3600$ & $2.2$ & $0.5$ & $0.4$\\
$320$ & $>3600$  & $12.3$ & $0.6$ & $0.5$\\
$640$ & $>3600$ & $36.8$ & $1.0$ & $0.8$\\
$1280$ & $>3600$ & $241.1$ & $22.6$ & $19.7$\\
$2560$ & $>3600$ & $>3600$ & $215.7$ & $155.2$\\
$5120$ & $>3600$ & $>3600$ & $1873.2$ & $1042.6$\\
\hline

\end{tabular}

\vspace{-.75cm}
\end{wraptable}

Symmetry breaking for QBF has already been studied more than ten years
ago~\cite{DBLP:conf/sat/AudemardMS04,DBLP:conf/ijcai/AudemardJS07,audemard2007efficient},
and it can have a dramatic effect on the performance of QBF solvers. As an
extreme example, the instances of the KBKF benchmark
set~\cite{DBLP:journals/iandc/BuningKF95} are highly symmetric. For some problem
sizes~$n$, we applied the two configurations QRes (standard Q-resolution) and LD
(long-distance resolution) of the solver DepQBF~\cite{DBLP:conf/cade/LonsingE17}
to this benchmark set.  For LD it is known that it performs exponentially better
than QRes on the KBKF formulas~\cite{DBLP:conf/lpar/EglyLW13}. The table above
shows the runtimes of DepQBF without and with symmetry breaking
(SB). While QRes-DepQBF only solves two formulas without symmetry breaking, with
symmetry breaking it even outperforms LD-DepQBF. Also for the LD configuration,
the symmetry breaking formulas are beneficial.  While this is an extreme
example, symmetries appear not only in crafted formulas.  In fact, we found that
about 60\% of the benchmarks used in the recent edition of
QBFEval\footnote{\url{http://www.qbflib.org/qbfeval17}} have nontrivial
symmetries that could be exploited.

Our goal in this paper is to develop an explicit, uniform, and general theory for symmetries of QBFs.
The theory is developed from scratch, and we include detailed proofs of all theorems.
The pioneering work on QBF symmetries~\cite{DBLP:conf/sat/AudemardMS04,DBLP:conf/ijcai/AudemardJS07,audemard2007efficient}
largely consisted in translating the well-known
techniques from SAT to QBF. This is not trivial, as universal quantifiers require special treatment.
Since then, however, research on QBF symmetry breaking almost stagnated. We believe that more work is
necessary. For example, we have observed that universal symmetry
breakers concerning universal variables fail to work correctly in recent 
clause-and-cube-learning QBF solvers when compactly provided as cubes. 
Although the encoding of the symmetry breaker is provably correct in theory, 
it turns out to be incompatible with
pruning techniques like pure literal elimination
for which already the compatibility with learning is not obvious~\cite{pure04}. Problems occur, for example, in the KBKF formulas
mentioned above. (Of course, for the reported timings we have only used parts of the symmetry breaking formula
which are provably correct both in theory and in practice.)

We hope that the theory developed in this paper
will help to resuscitate the interest in symmetries for QBF, 
lead to a better understanding of the interplay between symmetry breaking and modern optimization techniques,
provide a starting point for translating recent progress made in SAT and CSP to the QBF world, and
produce special symmetry breaking formulas that better exploit the unique features of~QBF.

\section{Quantified Boolean Formulas}\label{sec:qbf}

Let $X=\{x_1,\dots,x_n\}$ be a finite set of propositional variables
and $\BF(X)$ be a set of \emph{Boolean formulas} over $X$.
The elements of $\BF(X)$ are well-formed objects built from the
variables of $X$, truth constants $\top$ (true) and $\bot$ (false),
as well as logical connectives according to a certain grammar.
For most of the paper, we will not need to be very specific about the
structure of the elements of~$\BF(X)$.
We assume a well-defined semantics for the logical connectives, i.e.,
for every $\phi\in\BF(X)$ and every assignment $\sigma\colon X\to\{\top,\bot\}$
there is a designated value $[\phi]_\sigma\in\{\top,\bot\}$ associated to $\phi$ and~$\sigma$. In particular, we use $\land$ (conjunction),
$\lor$ (disjunction), $\leftrightarrow$ (equivalence),
$\rightarrow$ (implication), $\oplus$ (xor), and $\neg$ (negation) with
their standard semantics for combining and negating formulas.
Two formulas $\phi, \psi \in \BF(X)$ are \emph{equivalent} if
for every assignment $\sigma\colon X\to\{\top,\bot\}$ we have
$[\phi]_\sigma = [\psi]_\sigma$.
We use lowercase Greek letters for Boolean formulas and assignments.

If $f\colon\BF(X)\to\BF(X)$ is a function and $\sigma\colon X\to\{\top,\bot\}$ is
an assignment, the assignment $f(\sigma)\colon X\to\{\top,\bot\}$ is defined through
$f(\sigma)(x)=[f(x)]_\sigma$ ($x\in X$).
A partial assignment is a function $\sigma\colon Y\to\{\top,\bot\}$ with~$Y\subseteq X$.
If $\sigma$ is such a partial assignment and $\phi\in\BF(X)$, then $[\phi]_\sigma$ is supposed to
be an element of $\BF(X\setminus Y)$ such that for every assignment $\tau\colon X\to\{\top,\bot\}$ with $\tau|_Y=\sigma$ we
have $[[\phi]_\sigma]_\tau=[\phi]_\tau$. For example, imagine that the formula $[\phi]_\sigma$ is obtained
from $\phi$ by replacing every variable $y\in Y$ by the truth value~$\sigma(y)$.

We use uppercase Greek letters to denote \emph{quantified Boolean formulas} (QBFs).
A QBF has the form
$\Phi=P.\phi$ where $\phi\in\BF(X)$ is a Boolean formula and $P$ is a quantifier prefix for~$X$,
i.e., $P=Q_1x_1Q_2x_2\dots Q_nx_n$ for $Q_1,\dots,Q_n\in\{\forall,\exists\}$.
We only consider closed formulas, i.e., each element of $X$ appears in the prefix.
For a fixed prefix $P=Q_1x_1Q_2x_2\dots Q_nx_n$,
the \emph{quantifier block} of the variable $x_i$ is defined by
the smallest $i_{\min}\in\{1,\dots,i\}$ and the largest $i_{\max}\in\{i,\dots,n\}$ such that $Q_{i_{\min}}=\cdots=Q_{i_{\max}}$.

Every QBF is either true or false. The truth value is defined recursively as follows:
$\forall xP.\phi$ is true iff both $P.[\phi]_{\{x=\top\}}$ and $P.[\phi]_{\{x=\bot\}}$ are true,
and $\exists xP.\phi$ is true iff $P.[\phi]_{\{x=\top\}}$ or $P.[\phi]_{\{x=\bot\}}$ is true.
For example, $\forall x_1\exists x_2.(x_1 \leftrightarrow x_2)$ is true and $\exists x_1\forall x_2.(x_1\leftrightarrow x_2)$ is false.
The semantics of a QBF $P.\phi$ can also be described as a game for two players~\cite{DBLP:books/daglib/0072413}:
In the $i$th move, the truth value of $x_i$ is chosen by the existential player if $Q_i=\exists$ and by
the universal player if $Q_i=\forall$. The existential player wins if the resulting formula is true and the
universal player wins if the resulting formula is false. In this interpretation, a QBF is true if there is
a winning strategy for the existential player and it is false if there is a winning strategy for the universal player.

Strategies can be described as trees. Let
$P=Q_1x_1Q_2x_2\dots Q_nx_n$ be a prefix. An \emph{existential strategy} for $P$ is a tree of
height $n+1$ where every node at level $k\in\{1,\dots,n\}$ has one child if
$Q_k=\exists$ and two children if $Q_k=\forall$. In the case $Q_k=\forall$, the
two edges to the children are labeled by $\top$ and $\bot$, respectively. In
the case $Q_k=\exists$, the edge to the only child is labeled by either $\top$
or~$\bot$. \emph{Universal strategies} are defined analogously, the only difference being that the
roles of the quantifiers are exchanged, i.e., nodes at level $k$ have two
successors if $Q_k=\exists$ (one labeled $\bot$ and one labeled $\top$) and one
successor if $Q_k=\forall$ (labeled either $\bot$ or $\top$).
Here are the four existential strategies and the two universal strategies
for the prefix $\forall x_1\exists x_2$:
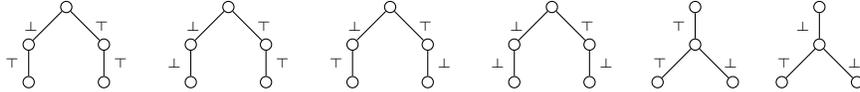
\begin{center}
  \begin{tikzpicture}[scale=.5]

    \begin{scope}[circle,inner sep=1.5pt]
    \node[draw] (root) at (0,0) {};
    \node[draw] (0) at (-1,-1) {};
    \node[draw] (1) at (1,-1) {};
    \node[draw] (01) at (-1,-2) {};
    \node[draw] (11) at (1,-2) {};
    \end{scope}

    \draw
    (root)edge node[left]{\tiny$\bot$} (0)
    (root)edge node[right]{\tiny$\top$} (1)
    (0)edge node[left]{\tiny$\top$} (01)
    (1)edge node[right]{\tiny$\top$} (11);
  \end{tikzpicture}\hfil
  \begin{tikzpicture}[scale=.5]

    \begin{scope}[circle,inner sep=1.5pt]
    \node[draw] (root) at (0,0) {};
    \node[draw] (0) at (-1,-1) {};
    \node[draw] (1) at (1,-1) {};
    \node[draw] (01) at (-1,-2) {};
    \node[draw] (11) at (1,-2) {};
    \end{scope}

    \draw
    (root)edge node[left]{\tiny$\bot$} (0)
    (root)edge node[right]{\tiny$\top$} (1)
    (0)edge node[left]{\tiny$\bot$} (01)
    (1)edge node[right]{\tiny$\top$} (11);
  \end{tikzpicture}\hfil
  \begin{tikzpicture}[scale=.5]

    \begin{scope}[circle,inner sep=1.5pt]
    \node[draw] (root) at (0,0) {};
    \node[draw] (0) at (-1,-1) {};
    \node[draw] (1) at (1,-1) {};
    \node[draw] (01) at (-1,-2) {};
    \node[draw] (11) at (1,-2) {};
    \end{scope}

    \draw
    (root)edge node[left]{\tiny$\bot$} (0)
    (root)edge node[right]{\tiny$\top$} (1)
    (0)edge node[left]{\tiny$\top$} (01)
    (1)edge node[right]{\tiny$\bot$} (11);
  \end{tikzpicture}\hfil
  \begin{tikzpicture}[scale=.5]

    \begin{scope}[circle,inner sep=1.5pt]
    \node[draw] (root) at (0,0) {};
    \node[draw] (0) at (-1,-1) {};
    \node[draw] (1) at (1,-1) {};
    \node[draw] (01) at (-1,-2) {};
    \node[draw] (11) at (1,-2) {};
    \end{scope}

    \draw
    (root)edge node[left]{\tiny$\bot$} (0)
    (root)edge node[right]{\tiny$\top$} (1)
    (0)edge node[left]{\tiny$\bot$} (01)
    (1)edge node[right]{\tiny$\bot$} (11);
  \end{tikzpicture}\hfil
  \begin{tikzpicture}[scale=.5]

    \begin{scope}[circle,inner sep=1.5pt]
    \node[draw] (root) at (0,0) {};
    \node[draw] (0) at (0,-1) {};
    \node[draw] (01) at (-1,-2) {};
    \node[draw] (11) at (1,-2) {};
    \end{scope}

    \draw
    (root)edge node[left]{\tiny$\top$} (0)
    (0)edge node[left]{\tiny$\top$} (01)
    (0)edge node[right]{\tiny$\bot$} (11);
  \end{tikzpicture}\hfil
  \begin{tikzpicture}[scale=.5]

    \begin{scope}[circle,inner sep=1.5pt]
    \node[draw] (root) at (0,0) {};
    \node[draw] (0) at (0,-1) {};
    \node[draw] (01) at (-1,-2) {};
    \node[draw] (11) at (1,-2) {};
    \end{scope}

    \draw
    (root)edge node[left]{\tiny$\bot$} (0)
    (0)edge node[left]{\tiny$\top$} (01)
    (0)edge node[right]{\tiny$\bot$} (11);
  \end{tikzpicture}
\end{center}
We write $\strat_\exists(P)$ for the set of all existential strategies and
$\strat_\forall(P)$ for the set of all universal strategies.
As shown in the following lemma, existential and universal strategies for the
same prefix a share at least one common path.
Unless stated otherwise, a path is meant to be complete in the sense that it
starts at the root and ends at a leaf. 
\begin{lem}
  If $P$ is a prefix and $s\in\strat_\exists(P)$, $t\in\strat_\forall(P)$,
  then $s$ and $t$ have a path in common.
\label{lem:mod}
\end{lem}
\begin{proof}
  A common path can be constructed by induction on the length of the prefix. There is nothing to show for
  prefixes of length~$0$. Suppose the claim holds for all prefixes of length~$n$ and consider a prefix
  $P'=P\,Q_{n+1}x_{n+1}$ of length $n+1$. Let $s\in\strat_\exists(P')$, $t\in\strat_\forall(P')$ be arbitrary.
  By chopping off the leafs of $s$ and~$t$, we obtain elements of $\strat_\exists(P)$ and $\strat_\forall(P')$,
  respectively, and these share a common path $\sigma_0$ by induction hypothesis.
  If $Q_{n+1}=\exists$, then $\sigma_0$ has a unique continuation in~$s$, with an edge labeled either $\top$ or~$\bot$,
  and $\sigma_0$ has two continuations in~$t$, one labeled $\top$ and one labeled~$\bot$, so the continuation of
  $\sigma_0$ in $s$ must also appear in~$t$.
  If $Q_{n+1}=\forall$, the argumentation is analogous.
  \qed
\end{proof}
Every path in a strategy for a prefix $P$ corresponds to an
assignment $\sigma\colon X\to\{\top,\bot\}$.
An existential strategy for QBF $P.\phi$ is a \emph{winning strategy} (for the existential player)
if all its paths are assignments for which $\phi$ is true.
A universal strategy is a \emph{winning strategy} (for the universal player)
if all its paths are assignments for which $\phi$ is false.
For a QBF $P.\phi$ and an existential strategy $s\in\strat_\exists(P)$, we
define $[P.\phi]_s=\bigwedge_\sigma[\phi]_\sigma$, where $\sigma$ ranges over
all the assignments corresponding to a path of~$s$.
Then we have $[P.\phi]_s=\top$ if and only if $s$ is an existential winning strategy.
For a universal strategy $t\in\strat_\forall(P)$, we define
$[P.\phi]_t=\bigvee_\tau[\phi]_\tau$, where $\tau$ ranges over all the assignments
corresponding to a path of~$t$.
Then $[P.\phi]_s=\bot$ if and only if $t$ is a universal winning strategy.

The definitions made in the previous paragraph are
consistent with the interpretation of QBFs introduced earlier: a QBF is true if and only if
there is an existential winning strategy, and it is false if and only if there is a universal winning strategy.
Lemma~\ref{lem:mod} ensures that a QBF is either true or false.
As another consequence of Lemma~\ref{lem:mod}, observe that for every QBF $P.\phi$ we have
\begin{alignat*}1
&\exists\ s\in\strat_\exists(P) : [P.\phi]_s=\top
\iff
\forall\ t\in\strat_\forall(P) : [P.\phi]_t=\top\\
\text{and}\quad&
\forall\ s\in\strat_\exists(P) : [P.\phi]_s=\bot
\iff
\exists\ t\in\strat_\forall(P) : [P.\phi]_t=\bot.
\end{alignat*}
We will also need the following property, the proof of which is straightforward.
\begin{lem}\label{lem:thread}
  Let $P$ be a prefix for~$X$, and let $\phi,\psi\in\BF(X)$.
  Then for all $s\in\strat_\exists(P)$ we have $[P.(\phi\land\psi)]_s=[P.\phi]_s\land[P.\psi]_s$,
  and for all $t\in\strat_\forall(P)$ we have $[P.(\phi\lor\psi)]_t=[P.\phi]_t\lor[P.\psi]_t$.
\end{lem}

\section{Groups and Group Actions}\label{sec:groups}

Symmetries can be described using groups and group actions~\cite{artin2011algebra}. Recall that a group is a set $G$ together with an
associative binary operation $G\times G\to G$, $(g,h)\mapsto gh$.
A group has a neutral element and every element
has an inverse in~$G$. A typical example for a group is the set $\set Z$ of integers together with addition. Another
example is the group of permutations. For any fixed $n\in\set N$, a permutation is a bijective
function $\pi\colon\{1,\dots,n\}\to\{1,\dots,n\}$. The set of all such functions together with composition forms a
group, called the symmetric group and denoted by~$S_n$.

A (nonempty) subset $H$ of a group~$G$ is called a subgroup of~$G$ if it is closed under the group operation and taking
inverses. For example, the set $2\set Z$ of all even integers is a subgroup of~$\set Z$, and the set
$\{\id,\binom{1\ 2\ 3}{1\ 3\ 2}\}$ is a subgroup of~$S_3$.
In general, a subset $E$ of $G$ is not a subgroup. However, for every subset $E$ we can consider the intersection
of all subgroups of $G$ containing~$E$. This is a subgroup and it is denoted by $\<E>$. The elements of $E$
are called \emph{generators} of the subgroup.
For example, we have $2\set Z=\<2>$, but also $2\set Z=\<4,6>$.
A set of generators for $S_3$ is $\{\binom{1\ 2\ 3}{2\ 3\ 1}, \binom{1\ 2\ 3}{2\ 1\ 3}\}$.

If $G$ is a group and $S$ is a set, then a \emph{group action} is a map $G\times S\to S$, $(g,x)\mapsto g(x)$
which is compatible with the group operation in the sense that for all $g,h\in G$ and $x\in S$ we have $(gh)(x)=g(h(x))$
and $e(x)=x$, where $e$ is the neutral element of~$G$. Note that when we have a group action, every element of $G$ can
be viewed as a bijective function~$S\to S$.

For example, for $G=S_n$ and $S=\{1,\dots,n\}$ we have a group action by the definition of the elements of~$S_n$.
Alternatively, we can let $S_n$ act on a set of tuples of length~$n$, say on $S=\{\Box,\blacksquare,\triangle\}^n$,
via permutation of the indices, i.e., $\pi(x_1,\dots,x_n) = (x_{\pi(1)},\dots,x_{\pi(n)})$. For example, for
$g=\binom{1\ 2\ 3}{1\ 3\ 2}$ we would have $g(\Box,\blacksquare,\Box)=(\Box,\Box,\blacksquare)$,
$g(\triangle,\triangle,\Box)=g(\triangle,\Box,\triangle)$,
$g(\blacksquare,\triangle,\triangle)=(\blacksquare,\triangle,\triangle)$, etc.
As one more example, we can consider the group $G=S_n\times S_m$ consisting of all pairs of permutations. The operation
for this group is defined componentwise, i.e., $(\pi,\sigma)(\pi',\sigma')=(\pi\pi',\sigma\sigma')$.
We can let $G$ act on a set of two dimensional arrays with shape $n\times m$, say on
$S=\{\Box,\blacksquare,\triangle\}^{n\times m}$, by letting the first component of a group element permute the row
index and the second component permute the column index.     For example, for $g=(\binom{1\ 2\ 3}{1\ 3\ 2}, \binom{1\ 2\ 3}{2\ 3\ 1})$ we then have
    \[
    g(\ \begin{tabular}{|c|c|c|}\hline
       $\Box$&$\blacksquare$&$\triangle$\\\hline
       $\Box$&$\triangle$&$\Box$\\\hline
       $\blacksquare$&$\blacksquare$&$\Box$\\\hline
    \end{tabular}\ ) =
    \begin{tabular}{|c|c|c|}\hline
       $\blacksquare$&$\triangle$&$\Box$\\\hline
       $\blacksquare$&$\Box$&$\blacksquare$\\\hline
       $\triangle$&$\Box$&$\Box$\\\hline
    \end{tabular}.
    \]
If we have a group action $G\times S\to S$, we can define an equivalence relation on $S$ via $x\sim y\iff\exists\ g\in G: x=g(y)$.
The axioms of groups and group actions ensure that $\sim$ is indeed an equivalence relation.
The equivalence classes are called the \emph{orbits} of the group action.
For example, for the action of $S_3$ on $\{\Box,\blacksquare,\triangle\}^3$ discussed above, there are
some orbits of size~1 (e.g., $\{(\blacksquare,\blacksquare,\blacksquare)\}$),
some orbits of size~3 (e.g., $\{(\Box,\Box,\triangle),(\Box,\triangle,\Box),(\triangle,\Box,\Box)\}$),
and there is one orbit of size~6 ($\{
(\Box,\blacksquare,\triangle),\penalty0
(\Box,\triangle,\blacksquare),\penalty0
(\blacksquare,\triangle,\Box),\penalty0
(\blacksquare,\Box,\triangle),\penalty0
(\triangle,\blacksquare,\Box),\penalty0
(\triangle,\Box,\blacksquare)\}$).

\section{Syntactic Symmetries}

We use group actions to describe symmetries of QBFs.
Two kinds of group actions are of interest.
On the one hand, we consider transformations that map formulas to formulas, i.e., a group action $G\times\BF(X)\to\BF(X)$.
On the other hand, we consider transformations that map strategies to strategies, i.e., a group action
$G\times\strat_\exists(P)\to\strat_\exists(P)$ or $G\times\strat_\forall(P)\to\strat_\forall(P)$.
In both cases, we consider groups $G$ which preserve the set of
winning strategies for a given QBF $P.\phi$.

Let us first consider group actions $G\times\BF(X)\to\BF(X)$.
In this case, we need to impose a technical restriction introduced in the following definition.
\begin{defi}\label{def:adm}
  Let $P$ be a prefix for~$X$.
  A bijective function $f\colon\BF(X)\to\BF(X)$ is called \emph{admissible} (w.r.t.~$P$) if
  \begin{enumerate}
  \item for every assignment $\sigma\colon X\to\{\top,\bot\}$ and every
    formula $\phi\in\BF(X)$ we have $[\phi]_{f(\sigma)} = [f(\phi)]_\sigma$;
  \item for every variable $x \in X$ the formula $f(x)$ only contains variables
    that belong to the same quantifier block of $P$ as~$x$.
  \end{enumerate}
\end{defi}
The first condition ensures that an admissible function $f$ preserves
propositional satisfiability. In particular, it implies that
for any $\phi, \psi \in \BF(X)$, the formulas $f(\neg\phi)$ and $\neg f(\phi)$ are equivalent, as are
$f(\phi \circ \psi)$ and $f(\phi) \circ f(\psi)$ for every binary connective~$\circ$.
As a consequence, it follows that the inverse of an admissible function is again admissible.
It also follows that an admissible function $f$ is essentially determined by its values for the variables.
Note however that variables can be mapped to arbitrary formulas.
\begin{ex}
  Let $X=\{x,y,a,b\}$ and $P=\forall x \forall y\exists a\exists b$.
  There is an admissible function $f$ with
  $f(x) = \neg x, f(y) = y, f(a) = b, f(b) = a$.
  For such a function, we may have $f(x\lor(a\to y))=\neg x\lor(b\to y)$.
  A function $g$ with $g(x)=a$ cannot be admissible,
  because of the second condition.
  By the first condition, a function $h$ with $h(x)=x$ and
  $h(y)=\neg x$ cannot be admissible.
\end{ex}
Next we show that admissible functions not only preserve satisfiability of
Boolean formulas, but also the truth of QBFs.
\begin{thm}\label{thm:admissible}
Let $P$ be a prefix for $X$ and $f\colon\BF(X)\to\BF(X)$ be
admissible for~$P$.
For any $\phi\in\BF(X)$ the formula $P.\phi$ is true if and only
if $P.f(\phi)$ is true.
\end{thm}
\begin{proof}
  Since the inverse of an admissible function is admissible, it suffices to show ``$\Rightarrow$''.
  To do so, we proceed by induction on the number of quantifier blocks in~$P$.

  There is nothing to show when $P$ is empty.
  Suppose the claim is true for all prefixes with $k$ quantifier blocks, and consider a prefix
  $P=Qx_1Qx_2\cdots Qx_iP'$ for some $i\in\{1,\dots,n\}$, $Q\in\{\forall,\exists\}$, and a prefix $P'$
  for $x_{i+1},\dots,x_n$ with at most $k$ quantifier blocks whose top quantifier is not~$Q$.
  By the admissibility, we may view $f$ as a pair of functions $f_1\colon\BF(\{x_1,\dots,x_i\})\to\BF(\{x_1,\dots,x_i\})$
  and $f_2\colon\BF(\{x_{i+1},\dots,x_n\})\to\BF(\{x_{i+1},\dots,x_n\})$, where $f_2$ is admissible for~$P'$.
  Let $s\in\strat_\exists(P)$ be a winning strategy for $P.\phi$.
  We will construct a winning strategy $t\in\strat_\exists(P)$ for $P.f(\phi)$.

  Case~1: $Q=\exists$. In this case, the upper $i$ levels of $s$ and $t$ consist of single paths.
  Let $\sigma\colon\{x_1,\dots,x_i\}\to\{\top,\bot\}$ be the assignment corresponding to the upper $i$ levels of~$s$.
  The subtree $s_\sigma$ of $s$ rooted at
  the end of $\sigma$ (level $i+1$) is a winning strategy for $P'.[\phi]_\sigma$.
  By induction hypothesis, $P'.f_2([\phi]_\sigma)$ has a winning strategy.
  Let $t$ have an initial path corresponding to the assignment $\tau=f_1^{-1}(\sigma)$ followed by a winning strategy
  of $P'.f_2([\phi]_\sigma)$. (Since $f_1$ is invertible and independent of $x_{i+1},\dots,x_n$, the
  assignment $\tau$ is well-defined.)
  Then $t$ is a winning strategy of $P.f(\phi)$.
  To see this, let $\rho$ be an arbitrary path of~$t$. We show that $[f(\phi)]_\rho=\top$.
  Indeed,
  \begin{alignat*}1
    [f(\phi)]_\rho &
    \overset{\vbox{\clap{\scriptsize\strut $t$ starts with $\tau$}\kern-5pt\clap{$\downarrow$}}} =
    [[f(\phi)]_\tau]_\rho
    \overset{\vbox{\clap{\scriptsize\strut Def. of $\tau$}\kern-5pt\clap{$\downarrow$}}} =
    [[f(\phi)]_{f_1^{-1}(\sigma)}]_\rho
    \overset{\vbox{\clap{\scriptsize\strut Def. of $f_1,f_2$}\kern-5pt\clap{$\downarrow$}}} =
    [[f_1(f_2(\phi))]_{f_1^{-1}(\sigma)}]_\rho \\[3pt]&
    \underset{\vbox{\clap{$\uparrow$}\kern-5pt\clap{\scriptsize\strut $f_1$ admissible}}} =
    [[f_1^{-1}(f_1(f_2(\phi)))]_{\sigma}]_\rho
    =
    [[f_2(\phi)]_{\sigma}]_\rho
    \underset{\vbox{\clap{$\uparrow$}\kern-5pt\clap{\scriptsize\strut $f_2$ admissible}}} =
    [f_2([\phi]_\sigma)]_\rho
    \underset{\vbox{\clap{$\uparrow$}\kern-5pt\clap{\scriptsize\strut choice of $t$}}} =
    \top.
  \end{alignat*}

  Case~2: $Q=\forall$. In this case, the upper $i$ levels of both $s$ and $t$ form complete binary trees in which every path
  corresponds to an assignment for the variables $x_1,\dots,x_i$.
  Let $\tau\colon\{x_1,\dots,x_i\}\to\{\top,\bot\}$ be such an assignment, and let
  $\sigma=f_1(\tau)$. 
  Let $s_\sigma$ be the subtree of~$s$ rooted at~$\sigma$.
  This is a winning strategy for the formula $P'.[\phi]_\sigma$ obtained from $P.\phi$ by instantiating the
  variables $x_1,\dots,x_i$ according to $\sigma$ and dropping the corresponding part of the prefix.
  By induction hypothesis, $P'.f_2([\phi]_\sigma)$ has a winning strategy.
  Pick one and use it as the subtree of $t$ rooted at~$\tau$.
  The same calculation as in Case~1 shows that $t$ is a winning strategy for $P.f(\phi)$.
  \qed
\end{proof}
Next we introduce the concept of a \emph{syntactic symmetry group.}
The attribute `syntactic' shall emphasize that this group acts
on formulas, in contrast to the `semantic' symmetry group introduced later,
which acts on strategies. Our distinction between syntactic and semantic 
symmetries corresponds to the distinction between the problem 
and solution symmetries made in CSP~\cite{DBLP:conf/aaai/CohenJJPS06}.
\begin{defi}\label{def:syn}
  Let $P.\phi$ be a QBF and let $G\times\BF(X)\to\BF(X)$ be a group action such that every $g\in G$ is admissible w.r.t.~$P$. 
  We call $G$ a \emph{syntactic symmetry group} for $P.\phi$ if $\phi$ and $g(\phi)$ are equivalent for all $g\in G$.
\end{defi}
It should be noticed that being a `symmetry group'
is strictly speaking not a property of the group itself but rather a property of the action
of $G$ on~$\BF(X)$. Further, we call a group action admissible if 
every $g \in G$ is admissible and we call $g \in G$ a \emph{symmetry}. 
Definition~\ref{def:syn} implies that when $G$ is a syntactic symmetry group for $P.\phi$, then for every
element $g\in G$ the QBF $P.g(\phi)$ has the same set of winning strategies as~$P.\phi$.
Note that this is not already a consequence of Thm.~\ref{thm:admissible}, which only said that $P.g(\phi)$ is true
if and only if $P.\phi$ is true, which does not imply that they have the same winning strategies.
\begin{ex}
  Consider the QBF $\Phi=P.\phi=\forall x\forall y\exists a\exists b. ((x\leftrightarrow a)\land(y\leftrightarrow b))$.
  A syntactic symmetry group for $\Phi$ is $G=\{\id,f\}$, where $f$ is an admissible function with
  $f(x)=y$, $f(y)=x$, $f(a)=b$, $f(b)=a$.

  Symmetries are often restricted to functions which map variables to literals. But this restriction
  is not necessary. Also the admissible function $g$ defined by $g(x)=x$, $g(y)=x\oplus y$, $g(a)=a$, $g(b)=a\oplus b$
  is a syntactic symmetry for~$\Phi$. 
\end{ex}

\section{Semantic Symmetries}

For the definition of semantic symmetry groups, no technical requirement like the admissibility is needed.
Every permutation of strategies that maps winning strategies to winning strategies is fine.
\begin{defi}\label{def:semg}
  Let $\Phi=P.\phi$ be a QBF and let $G$ be a group acting on $\strat_\exists(P)$ (or on $\strat_\forall(P)$).
  We call $G$ a \emph{semantic symmetry group} for $\Phi$ if for all $g\in G$ and
  all $s\in\strat_\exists(P)$ (or all $s\in\strat_\forall(P)$) we have $[\Phi]_s=[\Phi]_{g(s)}$.
\end{defi}
A single syntactic symmetry can give rise to several distinct semantic symmetries, as shown in the following example.
\begin{ex}\label{ex:7}
  Consider again $\Phi=P.\phi=\forall x\forall y\exists a\exists b. ((x\leftrightarrow a)\land(y\leftrightarrow b))$.
  The function $f$ of the previous example, which exchanges $x$ with $y$ and $a$ with~$b$ in the formula,
  can be translated to a semantic symmetry~$\tilde f$:
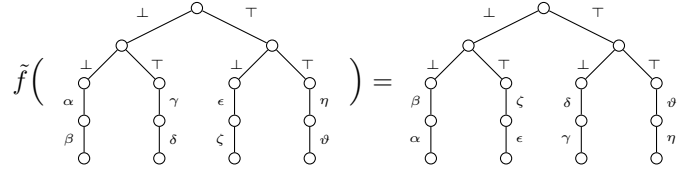
  \[
    \tilde f\Bigl(\ \vcenter{\hbox{\begin{tikzpicture}[scale=.5]

    \begin{scope}[circle,inner sep=1.5pt]
    \node[draw] (root) at (0,0) {};
    \node[draw] (0) at (-2,-1) {};
    \node[draw] (1) at (2,-1) {};
    \node[draw] (10) at (-3,-2) {};
    \node[draw] (11) at (-1,-2) {};
    \node[draw] (00) at (3,-2) {};
    \node[draw] (01) at (1,-2) {};
    \node[draw] (100) at (-3,-3) {};
    \node[draw] (110) at (-1,-3) {};
    \node[draw] (000) at (3,-3) {};
    \node[draw] (010) at (1,-3) {};
    \node[draw] (1000) at (-3,-4) {};
    \node[draw] (1100) at (-1,-4) {};
    \node[draw] (0000) at (3,-4) {};
    \node[draw] (0100) at (1,-4) {};
    \end{scope}

    \draw
    (root)edge node[above left]{\tiny$\bot$} (0)
    (root)edge node[above right]{\tiny$\top$} (1)
    (0)edge node[left]{\tiny$\bot$} (10)
    (0)edge node[right]{\tiny$\top$} (11)
    (1)edge node[left]{\tiny$\bot$} (01)
    (1)edge node[right]{\tiny$\top$} (00)
    (10)edge node[left]{\tiny$\alpha$} (100)
    (11)edge node[right]{\tiny$\gamma$} (110)
    (01)edge node[left]{\tiny$\epsilon$} (010)
    (00)edge node[right]{\tiny$\eta$} (000)
    (100)edge node[left]{\tiny$\beta$} (1000)
    (110)edge node[right]{\tiny$\delta$} (1100)
    (010)edge node[left]{\tiny$\zeta$} (0100)
    (000)edge node[right]{\tiny$\vartheta$} (0000);
    \end{tikzpicture}}}\ \Bigr)=
    \vcenter{\hbox{%
    \begin{tikzpicture}[scale=.5,baseline={(0,-2)}]

    \begin{scope}[circle,inner sep=1.5pt]
    \node[draw] (root) at (0,0) {};
    \node[draw] (0) at (-2,-1) {};
    \node[draw] (1) at (2,-1) {};
    \node[draw] (10) at (-3,-2) {};
    \node[draw] (11) at (-1,-2) {};
    \node[draw] (00) at (3,-2) {};
    \node[draw] (01) at (1,-2) {};
    \node[draw] (100) at (-3,-3) {};
    \node[draw] (110) at (-1,-3) {};
    \node[draw] (000) at (3,-3) {};
    \node[draw] (010) at (1,-3) {};
    \node[draw] (1000) at (-3,-4) {};
    \node[draw] (1100) at (-1,-4) {};
    \node[draw] (0000) at (3,-4) {};
    \node[draw] (0100) at (1,-4) {};
    \end{scope}

    \draw
    (root)edge node[above left]{\tiny$\bot$} (0)
    (root)edge node[above right]{\tiny$\top$} (1)
    (0)edge node[left]{\tiny$\bot$} (10)
    (0)edge node[right]{\tiny$\top$} (11)
    (1)edge node[left]{\tiny$\bot$} (01)
    (1)edge node[right]{\tiny$\top$} (00)
    (10)edge node[left]{\tiny$\beta$} (100)
    (11)edge node[right]{\tiny$\zeta$} (110)
    (01)edge node[left]{\tiny$\delta$} (010)
    (00)edge node[right]{\tiny$\vartheta$} (000)
    (100)edge node[left]{\tiny$\alpha$} (1000)
    (110)edge node[right]{\tiny$\epsilon$} (1100)
    (010)edge node[left]{\tiny$\gamma$} (0100)
    (000)edge node[right]{\tiny$\eta$} (0000);
    \end{tikzpicture}}}
  \]
  This symmetry exchanges 
  the labels of level 3 and level 4 and swaps the existential parts of the 
  two paths in the middle.
  Regardless of the choice of $\alpha,\dots,\eta\in\{\bot,\top\}$, the strategy on the left is winning if and only if
  the strategy on the right is winning, so $\tilde f$ maps winning strategies to winning strategies.

  Some further semantic symmetries can be constructed from~$f$.
  For example, in order to be a winning strategy, it is necessary that $\alpha=\beta=\bot$.
  So we can take a function that just flips $\alpha$ and $\beta$ but does not touch the rest of the tree.
  For the same reason, also a function that just flips $\eta$ and $\vartheta$ but does not affect the rest of the
  tree is a semantic symmetry.
  The composition of these two functions and the function $\tilde f$ described before (in an arbitrary order)
  yields a symmetry that exchanges $\gamma$ with $\zeta$ and $\delta$ with $\epsilon$ but keeps $\alpha,\beta,\eta,\vartheta$
  fixed. Also this function is a semantic symmetry.
\end{ex}
The construction described in the example above works in general.
Recall that for an assignment
$\sigma\colon X\to\{\top,\bot\}$ and a function $f\colon \BF(X)\to\BF(X)$,
the assignment $f(\sigma)\colon X\to\{\top,\bot\}$ is defined by
$f(\sigma)(x) = [f(x)]_\sigma$ for $x\in X$.
\begin{lem}\label{lem:p}
  Let $P$ be a prefix for $X$ and
  $g$ be an element of a group acting admissibly on $\BF(X)$.
  Then there is a function
  $f\colon\strat_\exists(P)\to\strat_\exists(P)$ such that for all $s\in\strat_\exists(P)$
  we have that $\sigma$ is a path of $f(s)$ if and only if $g(\sigma)$ is a path of~$s$.
\end{lem}
\begin{proof}
  Since $g$ is an admissible function, it acts independently on variables belonging to different
  quantifier blocks. Therefore it suffices to consider the case where $P$ consists of a single
  quantifier block. If all quantifiers are existential, then $s$ consists of a single path,
  so the claim is obvious. If there are only universal quantifiers, then $s$ consists of a complete binary
  tree containing all possible paths, so the claim is obvious as well. \qed
\end{proof}
Starting from a syntactic symmetry group~$G\syn$, we can consider all the semantic
symmetries that can be obtained from it like in the example above. All these
semantic symmetries from a semantic symmetry group, which we call the semantic
symmetry group associated to~$G\syn$.
\begin{defi}\label{def:1}
  Let $P$ be a prefix for~$X$ and let $G\syn\times\BF(X)\to\BF(X)$ be an admissible group action.
  Let $G\sem$ be the set of all bijective functions $f\colon\strat_\exists(P)\to\strat_\exists(P)$ such that
  for all $s\in\strat_\exists(P)$ and every path $\sigma$ of $f(s)$ there exists a $g\in G\syn$
  such that $g(\sigma)$ is a path of~$s$.
  This $G\sem$ is called the \emph{associated group} of~$G\syn$.
\end{defi}
Again, it would be formally more accurate but less convenient to say that the action of $G\sem$ on $\strat_\exists(P)$ is
associated to the action of $G\syn$ on~$\BF(X)$.
\begin{thm}
  If $G\syn$ is a syntactic symmetry group for a QBF~$\Phi$,
  then the associated group $G\sem$ of $G\syn$ is a semantic symmetry group for~$\Phi$.
\end{thm}
\begin{proof}
  Let $\Phi=P.\phi$.
  Obviously, $G\sem$ is a group. To show that it is a symmetry group,
  let $s\in\strat_\exists(P)$ be a winning strategy for~$\Phi$, and let $g\sem\in G\sem$.
  We show that $g\sem(s)$ is again a winning strategy.
  Let $\sigma$ be a path of~$g\sem(s)$.
  By Def.~\ref{def:1},
   there exists a $g\syn\in G\syn$ such that $g\syn(\sigma)$ is a path of~$s$.
  Since $s$ is a winning strategy, $[\phi]_{g\syn(\sigma)}=\top$, and since $G\syn$ is a symmetry group,
  $[\phi]_{g\syn(\sigma)}=[g\syn(\phi)]_{g\syn(\sigma)}$.
  By admissibility $[g\syn(\phi)]_{g\syn(\sigma)}=[\phi]_\sigma=\top$.
  Hence every path of $g\sem(s)$ is a satisfying assignment, so $g\sem(s)$ is a winning
  strategy. \qed
\end{proof}
The distinction between a syntactic and a semantic symmetry groups is immaterial when the
prefix consists of a single quantifier block. In particular, SAT problems can be
viewed as QBFs in which all quantifiers are~$\exists$. For such formulas, each tree in $\strat_\exists(P)$
consists of a single paths, so in this case the requirement $\forall\ s\in\strat_\exists(P):[\Phi]_s=[\Phi]_{g(s)}$
from Def.~\ref{def:semg} boils down to the requirement that $[\phi]_\sigma=[\phi]_{f(\sigma)}$ should hold
for all assignments $\sigma\colon X\to\{\top,\bot\}$. This reflects the 
condition of Def.~\ref{def:syn} that
$\phi$ and $f(\phi)$ are equivalent. 

As we have seen in Example~\ref{ex:7}, there is more diversity for prefixes with several quantifier
blocks. In such cases, a single element of a syntactic symmetry group can give
rise to a lot of elements of the associated semantic symmetry group. In fact, the associated
semantic symmetry group is very versatile. For example, when there are
two strategies $s,s'\in\strat_\exists(P)$ and some element~$f$ of an associated semantic
symmetry group $G\sem$ such that $f(s)=s'$, then there is also an element $h\in G\sem$
with $h(s)=s'$, $h(s')=s$ and $h(r)=r$ for all $r\in\strat_\exists(P)\setminus\{s,s'\}$. The
next lemma is a generalization of this observation which indicates that $G\sem$ contains
elements that exchange subtrees across strategies.
\begin{lem}\label{lem:10}
  Let $P=Q_1x_1\dots Q_nx_n$
  be a prefix and $G\syn\times\BF(X)\to\BF(X)$ be an admissible group action.
  Let $G\sem$ be the associated group of~$G\syn$.
  Let $s\in\strat_\exists(P)$ and let $\sigma$ be a path of~$s$.
  Let $i\in\{1,\dots,n\}$ be such that $[x_j]_\sigma=[g(x_j)]_\sigma$ for all $g\in G\syn$
  and all $j<i$.

  Further, let $f\in G\sem$ and $s'=f(s)$. Let $\sigma'$ be a path of $s'$
  such that the first $i-1$
  edges of $\sigma'$ agree with the first $i-1$ edges of~$\sigma$.
  By the choice of $i$ such a $\sigma'$ exists.
  Let $t,t'\in\strat_\exists(Q_ix_i\dots Q_nx_n)$ be the subtrees of $s,s'$ rooted at
  the the $i$th node of $\sigma,\sigma'$, respectively, and
  let $s''\in\strat_\exists(P)$ be the strategy obtained from $s$ by replacing $t$ by~$t'$, as illustrated in the picture
  below.
  Then there exists $h\in G\sem$ with $h(s)=s''$.
\end{lem}
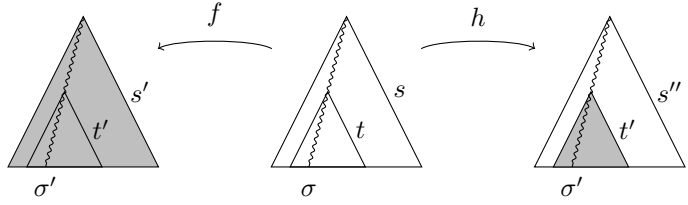
\begin{center}
  \begin{tikzpicture}[scale=.5]
    \node (s') at (0,-2) {};
    \draw[fill=lightgray] (0,0) -- (1,-2) node[right] {$s'$} -- (2,-4)--(-2,-4)--cycle;
    \draw[xshift=-.5cm] (0,-2)--(.5,-3)node[right] {$t'$} --(1,-4)--(-1,-4)--cycle;
    \draw (0,0)decorate[decoration={snake,amplitude=.2mm,segment length=1mm}] {--(-1,-4)}
     node[below] {$\mathstrut\sigma'$};
    \begin{scope}[xshift=7cm]
    \node (s) at (0,-2) {};
    \draw (0,0) -- (1,-2) node[right] {$s$} -- (2,-4)--(-2,-4)--cycle;
    \draw[xshift=-.5cm] (0,-2)--(.5,-3)node[right] {$t$} --(1,-4)--(-1,-4)--cycle;
    \draw (0,0)decorate[decoration={snake,amplitude=.2mm,segment length=1mm}] {--(-1,-4)}
     node[below] {$\mathstrut\sigma$};
    \end{scope}
    \begin{scope}[xshift=14cm]
    \node (s'') at (0,-2) {};
    \draw (0,0) -- (1,-2) node[right] {$s''$} -- (2,-4)--(-2,-4)--cycle;
    \draw[fill=lightgray,xshift=-.5cm] (0,-2)--(.5,-3)node[right] {$t'$} --(1,-4)--(-1,-4)--cycle;
    \draw (0,0)decorate[decoration={snake,amplitude=.2mm,segment length=1mm}] {--(-1,-4)}
    node[below] {$\mathstrut\sigma'$};
    \end{scope}
    \draw[shorten <=1cm,shorten >=1cm,->] (s) edge[bend right] node[above,yshift=2mm] {$f$} (s');
    \draw[shorten <=1cm,shorten >=1cm,->] (s) edge[bend left] node[above,yshift=2mm] {$h$} (s'');
  \end{tikzpicture}
\end{center}
\begin{proof}
  Define $h\colon\strat_\exists(P)\to\strat_\exists(P)$ by
  $h(s)=s''$, $h(s'')=s$, and $h(r)=r$ for all $r\in\strat_\exists(P)\setminus\{s,s''\}$.
  Obviously, $h$ is a bijective function from $\strat_\exists(P)$ to $\strat_\exists(P)$.
  To show that $h$ belongs to~$G\sem$, we must show that for every
  $r\in\strat_\exists(P)$ and every path $\rho$ of $h(r)$ there exists $g\in G\syn$ such that
  $g(\rho)$ is a path of~$r$.
  For $r\in\strat_\exists(P)\setminus\{s,s''\}$ we have $h(r)=r$, so there is nothing to show.

  Consider the case $r=s$. Let $\rho$ be a path of $h(r)=s''$. If $\rho$ does not end in the
  subtree~$t'$, then the same path $\rho$ also appears in $r$ and we can take $g=\id$.
  Now suppose that $\rho$ does end in the subtree~$t'$.
  Then $\rho$ is also a path of $s'=f(s)$, because all paths of $s$ and $s'$ ending in $t$ or $t'$
  agree above the $i$th node.
  Since $f\in G\sem$, there exists $g\in G\syn$ such that $g(\rho)$ is a path of~$s$.

  Finally, consider the case $r=s''$. Let $\rho$ be a path of $h(r)=s$. If $\rho$ does not end
  in the subtree~$t$, then the same path $\rho$ also appears in $r$ and we can take $g=\id$.
  Now suppose that $\rho$ does end in the subtree~$t$.
  Then the first $i-1$ edges of $\rho$ agree with those of~$\sigma$.
  Since $s=f^{-1}(s')$, there exists $g\in G\syn$ such that $g(\rho)$ is a path of~$s'$.
  By assumption on $G\syn$, the element $g$ fixes first $i-1$ edges of $\rho$, so
  $g(\rho)$ ends in~$t'$ and is therefore a path of~$s''$, as required.
  \qed
\end{proof}

\section{Existential Symmetry Breakers}

The action of a syntactic symmetry group of a QBF $P.\phi$ splits $\BF(X)$ into orbits.
For all the formulas $\psi$ in the orbit of $\phi$, the QBF $P.\psi$ has exactly the same winning
strategies as~$P.\phi$. For finding a winning strategy, we therefore have the freedom of exchanging
$\phi$ with any other formula in its orbit.

The action of a semantic symmetry group on $\strat_\exists(P)$ splits $\strat_\exists(P)$ into orbits.
In this case, every orbit either contains only winning strategies for $P.\phi$ or no winning strategies
for $P.\phi$ at all:
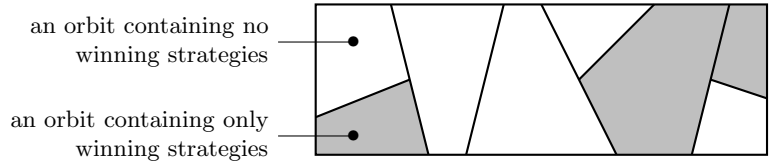
\begin{center}
  \begin{tikzpicture}[scale=.5]
    \fill[lightgray] (0,0)--(0,1)--(2.5,2)--(3,0)--cycle (8,0)--(7,2)--(9,4)--(12,4)--(12,1.5)--(10.5,2)--(10,0)--cycle;
    \draw[thick] (0,0) rectangle (12,4);
    \draw[thick] (0,1)--(2.5,2) (3,0) -- (2,4) (4,0)--(5,4) (6,4)--(8,0) (10,0)--(11,4) (12,1.5)--(10.5,2) (7,2)--(9,4);
    \draw (1,.5) node {$\bullet$} -- (-1,.5) node[left] {\parbox{4cm}{\raggedleft\footnotesize an orbit containing only winning strategies}};
    \draw (1,3) node {$\bullet$} -- (-1,3) node[left] {\parbox{4cm}{\raggedleft\footnotesize an orbit containing no winning strategies}};
  \end{tikzpicture}
\end{center}
Instead of checking all elements of $\strat_\exists(P)$, it is sufficient to check one element per orbit.
If a winning strategy exists, then any such sample contains one.

To avoid inspecting strategies that belong to the same orbit symmetry
breaking  
introduces a formula $\psi\in\BF(X)$ which is such that $P.\psi$ has at least one winning strategy in every
orbit. Such a formula is called a \emph{symmetry breaker}.
The key observation is that instead of solving $P.\phi$, we can solve $P.(\phi\land\psi)$.
Every winning strategy for the latter will be a winning strategy for the former, and if the former has
at least one winning strategy, then so does the latter.
By furthermore allowing transformations of $\phi$ via a syntactic symmetry group, we get the following definition.
\begin{defi}
  Let $P$ be a prefix for~$X$, let $G\syn$ be a group acting admissibly on $\BF(X)$ and let $G\sem$ be a group action on $\strat_\exists(P)$.
  A formula $\psi\in\BF(X)$ is called an \emph{existential symmetry breaker} for~$P$
  (w.r.t.\ the actions of $G\syn$ and $G\sem$)
  if for every $s\in\strat_\exists(P)$ there exist $g\syn\in G\syn$ and $g\sem\in G\sem$
  such that $[P.g\syn(\psi)]_{g\sem(s)}=\top$.
\end{defi}
\hangindent=-2cm\hangafter=1
\begin{ex}
  Consider the formula $\Phi=P.\phi=\forall x\exists y\exists z.(y\leftrightarrow z)$.
  All the elements of $\strat_\exists(P)$ have the form depicted on the right.
  As syntactic symmetries, we have the admissible functions
  $f,g\colon\BF(X)\to\BF(X)$ defined by
  $f(x)=x$, $f(y)=z$, $f(z)=y$, and
  $g(x)=x$, $g(y)=\neg y$, $g(z)=\neg z$,
  respectively, so we can take $G\syn=\<f,g>$
  as a syntactic symmetry group.
  \hfill\null\smash{\rlap{\hbox to20mm{\hss\quad
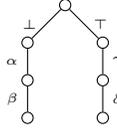
  \begin{tikzpicture}[scale=.5,baseline=-1.7cm]
    \begin{scope}[circle,inner sep=1.5pt]
    \node[draw] (root) at (0,0) {};
    \node[draw] (0) at (-1,-1) {};
    \node[draw] (1) at (1,-1) {};
    \node[draw] (01) at (-1,-2) {};
    \node[draw] (11) at (1,-2) {};
    \node[draw] (011) at (-1,-3) {};
    \node[draw] (111) at (1,-3) {};
    \end{scope}
    \draw
    (root)edge node[left]{\tiny$\bot$} (0)
    (root)edge node[right]{\tiny$\top$} (1)
    (0)edge node[left]{\tiny$\alpha$} (01)
    (1)edge node[right]{\tiny$\gamma$} (11)
    (01)edge node[left]{\tiny$\beta$} (011)
    (11)edge node[right]{\tiny$\delta$} (111);
  \end{tikzpicture}\hss}}}

  \hangindent=0pt
  According to standard techniques~\cite{audemard2007efficient} the formula
  $\lnot y$ is a symmetry breaker for $P.\phi$. 
  When considering $G\syn$ together with $G\sem=\{\id\}$ (what would be 
  sufficient for SAT), the complications for QBF become obvious.  
  The orbit of $\neg y$ is $O = \{y,z,\neg y,\neg z\}$. 
  Now consider the strategy with $\alpha=\top,\beta=\bot,\gamma=\bot,\delta=\top$. 
  For any $\psi \in O$, this strategy does not 
  satisfy $P.\psi$, because $\psi$ is true on one branch, but 
  false on the other. Using semantic symmetries can overcome 
  this problem. 

  Semantic symmetries can act differently on different paths.  Let
  $f_1\colon\strat_\exists(P)\to\strat_\exists(P)$ be the function which
  exchanges $\alpha,\beta$ and leaves $\gamma,\delta$ fixed, let
  $g_1\colon\strat_\exists(P)\to\strat_\exists(P)$ be the function which
  replaces $\alpha,\beta$ by $\neg\alpha,\neg\beta$ and leaves $\gamma,\delta$
  fixed, and let $f_2,g_2\colon\strat_\exists(P)\to\strat_\exists(P)$ be
  defined like $f_1,g_1$ but with the roles of $\alpha,\beta$ and
  $\gamma,\delta$ exchanged. The group $G\sem=\<f_1,g_1,f_2,g_2>$ is a semantic
  symmetry group for~$\Phi$. This group splits $\strat_\exists(P)$ into four orbits:
  one orbit consists of all strategies with $\alpha=\beta$, $\gamma=\delta$,
  one consists of those with $\alpha=\beta$, $\gamma\neq\delta$,
  one consists of those with $\alpha\neq\beta$, $\gamma=\delta$,
  and on consists of those with $\alpha\neq\beta$, $\gamma\neq\delta$.

  Taking $G\syn=\{\id\}$ together with this group~$G\sem$, the formula $\neg y$ is a symmetry breaker, because
  each orbit contains one element with $\alpha=\gamma=\bot$.
\end{ex}
\par\noindent
The following theorem is the main property of symmetry breakers.
\begin{thm}\label{thm:mainex}
  Let $\Phi=P.\phi$ be a QBF.
  Let $G\syn$ be a syntactic symmetry group and $G\sem$ be a semantic symmetry group acting on~$\strat_\exists(P)$.
  Let $\psi$ be an existential symmetry breaker for $G\syn$ and $G\sem$.
  Then $P.\phi$ is true iff $P.(\phi\land\psi)$ is true.
\end{thm}
\begin{proof}
  The direction ``$\Leftarrow$'' is obvious (by Lemma~\ref{lem:thread}). We show ``$\Rightarrow$''.
  Let $s\in\strat_\exists(P)$ be such that $[\Phi]_s=\top$. Since $\Phi$ is true, such an $s$ exists.
  Let $g\syn\in G\syn$ and $g\sem\in G\sem$ be such that such that $[P.g\syn(\psi)]_{g\sem(s)}=\top$.
  Since $\psi$ is an existential symmetry breaker, such elements exist.
  Since $G\syn$ and $G\sem$ are symmetry groups, $[P.g\syn(\phi)]_{g\sem(s)}=[P.\phi]_s=\top$.
  Lemma~\ref{lem:thread} implies $[P.(g\syn(\phi)\land g\syn(\psi))]_{g\sem(s)}=\top$.
  By the compatibility with logical operations (admissibility),
  \[
    [P.g\syn(\phi\land\psi)]_{g\sem(s)}=[P.(g\syn(\phi)\land g\syn(\psi))]_{g\sem(s)}=\top.
  \]
  Now by Thm.~\ref{thm:admissible} applied with $g\syn^{-1}$ to $P.g\syn(\phi\land\psi)$,
  it follows that there exists $s'$ such that $[P.(\phi\land\psi)]_{s'}=\top$, as claimed.  \qed
\end{proof}
As a corollary, we may remark that for an existential symmetry breaker $\psi$ for the prefix $P$
the formula $P.\psi$ is always true. To see this, choose $\phi=\top$ and observe that
any groups $G\syn$ and $G\sem$ are symmetry groups for~$\phi$. By the theorem, $P.(\phi\land\psi)$
is true, so $P.\psi$ is true.

\section{Universal Symmetry Breakers}

An inherent property of reasoning about QBFs is the duality between 
``existential'' and ``universal'' reasoning~\cite{DBLP:conf/sat/SabharwalAGHS06}, i.e., the duality 
between proving and refuting a QBF. For showing that a QBF is true, 
an existential strategy has to be found that is an existential winning 
strategy. An existential symmetry breaker tightens the pool of 
existential strategies among which the existential winning strategy 
can be found (in case there is one). 

If the given QBF is false, 
then a universal strategy has to be found that is a universal winning strategy.
In this case, an existential symmetry breaker is not useful.  
Recall that a universal winning strategy is a tree in which all paths are
falsifying assignments. Using an existential symmetry breaker as in Thm.~\ref{thm:mainex} tends to increase the
number of such paths and thus increases the number of potential candidates.
To aid the search for a universal winning strategy, it would be better to increase the number of paths
corresponding to satisfying assignments, because this reduces the search space for universal winning strategies.
For getting symmetry breakers serving this purpose, we can use a theory that is analogous to the theory of the
previous section.
\begin{defi}
  Let $P$ be a prefix for~$X$, let $G\syn$ be a group acting admissibly on $\BF(X)$ and let $G\sem$ be a group action on $\strat_\forall(P)$.
  A formula $\psi\in\BF(X)$ is called a \emph{universal symmetry breaker} for~$P$
  (w.r.t.\ the actions of $G\syn$ and $G\sem$)
  if for every $t\in\strat_\forall(P)$ there exist $g\syn\in G\syn$ and $g\sem\in G\sem$
  such that $[P.g\syn(\psi)]_{g\sem(t)}=\bot$.
\end{defi}
No change is needed for the definition of syntactic symmetry groups.
A semantic symmetry group for $\Phi=P.\phi$ is now a group acting on $\strat_\forall(P)$
in such a way that $[P.\phi]_t=[P.\phi]_{g(t)}$ for all $g\in G$ and all $t\in\strat_\forall(P)$.
With these adaptions, we have the following analog of Thm.~\ref{thm:mainex}.
\begin{thm}\label{thm:mainuniv}
  Let $\Phi=P.\phi$ be a QBF.
  Let $G\syn$ be a syntactic symmetry group and $G\sem$ be a semantic symmetry group acting on $\strat_\forall(P)$.
  Let $\psi$ be a universal symmetry breaker for $G\syn$ and~$G\sem$.
  Then $P.\phi$ is false iff $P.(\phi\lor\psi)$ is false.
\end{thm}
The proof is obtained from the proof of Thm.~\ref{thm:mainex} by replacing
$\strat_\exists(P)$ by $\strat_\forall(P)$, every $\land$ by~$\lor$, every
$\top$ by~$\bot$, and ``existential'' by ``universal''.

We have seen before that for an existential symmetry breaker $\psi_\exists$ the
QBF $P.\psi_\exists$ is necessarily true. Likewise, for a universal symmetry breaker $\psi_\forall$,
the QBF $P.\psi_\forall$ is necessarily false. This has the important consequence
that existential and universal symmetry breakers can be used in combination, even
if they are not defined with respect to the same group actions.
\begin{thm}
  Let $\Phi=P.\phi$ be a QBF.
  Let $G\syn^\exists$ and $G\syn^\forall$ be syntactic symmetry groups of~$\Phi$,
  let $G\sem^\exists$ be a semantic symmetry group of $\Phi$ acting on $\strat_\exists(P)$ and
  let $G\sem^\forall$ be a semantic symmetry group of $\Phi$ acting on $\strat_\forall(P)$.
  Let $\psi_\exists$ be an existential symmetry breaker for $G\syn^\exists$ and~$G\sem^\exists$, and
  let $\psi_\forall$ be an existential symmetry breaker for $G\syn^\forall$ and~$G\sem^\forall$.
  Then $P.\phi$ is true 
  iff $P.((\phi\lor\psi_\forall)\land\psi_\exists)$ is true
  iff $P.((\phi\land\psi_\exists)\lor\psi_\forall)$ is true.
\end{thm}
\begin{proof}
  \def\equiv#1{\stackrel{\hbox to 4em{\hss\scriptsize\strut #1\hss}}\iff}%
  For the first equivalence, we have
  \begin{alignat*}1
  \text{$P.\phi$ is true}
  &\equiv{Thm.~\ref{thm:mainuniv}} \text{$P.(\phi\lor\psi_\forall)$ is true}\\
  &\equiv{Def.} \exists\ s\in\strat_\exists(P): [P.(\phi\lor\psi_\forall)]_s=\top\\
  &\equiv{} \exists\ s\in\strat_\exists(P): [P.(\phi\lor\psi_\forall)]_s\land\underbrace{[P.\psi_\exists]_s}_{=\top}=\top\\
  &\equiv{Lem.~\ref{lem:thread}} \exists\ s\in\strat_\exists(P): [P.((\phi\lor\psi_\forall)\land \psi_\exists)]_s=\top\\
  &\equiv{Def.} \text{$P.((\phi\lor\psi_\forall)\land\psi_\exists)$ is true}.
  \end{alignat*}
  The proof of the second equivalence is analogous.  \qed
\end{proof}
Next we relate existential symmetry breakers to universal symmetry breakers.
Observe that when $P$ is a prefix and $\tilde P$ is the prefix obtained from $P$ by changing
all quantifiers, i.e., replacing each $\exists$ by $\forall$ and each $\forall$ by~$\exists$,
then $\strat_\exists(P)=\strat_\forall(\tilde P)$.
For any formula $\phi\in\BF(X)$ and any $s\in\strat_\exists(P)=\strat_\forall(\tilde P)$
we have $\neg[P.\phi]_s=[\tilde P.\neg\phi]_s$.
Therefore, if $G\syn$ is a group acting admissibly on $\BF(X)$ and $G\sem$ is a group acting
on $\strat_\exists(P)=\strat_\forall(\tilde P)$, we have
\begin{alignat*}1
  &\text{$\psi$ is an existential symmetry breaker for $G\syn$ and $G\sem$}\\
  &\iff \forall\ s\in\strat_\exists(P)\ \exists\ g\syn\in G\syn,g\sem\in G\sem: [P.g\syn(\psi)]_{g\sem(s)}=\top\\
  &\iff \forall\ s\in\strat_\forall(\tilde P)\ \exists\ g\syn\in G\syn,g\sem\in G\sem: [\tilde P.\neg g\syn(\psi)]_{g\sem(s)}=\bot\\
  &\iff \forall\ s\in\strat_\forall(\tilde P)\ \exists\ g\syn\in G\syn,g\sem\in G\sem: [\tilde P.g\syn(\neg\psi)]_{g\sem(s)}=\bot\\
  &\iff \text{$\neg\psi$ is universal symmetry breaker for $G\syn$ and $G\sem$},
\end{alignat*}
where admissibility of $g\syn$ is used in the third step. We have thus proven the following theorem which generalizes Property~2 of the symmetry breaker 
introduced in~\cite{audemard2007efficient}.
\begin{thm}\label{thm:duality}
  Let $P$ be a prefix for~$X$ and let $\tilde P$ be the prefix obtained from $P$ by flipping all the quantifiers.
  Let $G\syn$ be a group acting admissibly on~$\BF(X)$ and
  let $G\sem$ be a group acting on $\strat_\exists(P)=\strat_\forall(\tilde P)$.
  Then $\psi\in\BF(X)$ is an existential symmetry breaker for $G\syn$ and~$G\sem$ if and only if
  $\neg\psi$ is a universal symmetry breaker for $G\syn$ and~$G\sem$.
\end{thm}

\section{Construction of Symmetry Breakers}\label{sec:symqbf}

Because of Thm.~\ref{thm:duality}, it suffices to discuss the construction of existential symmetry breakers.
The universal symmetry breaker is obtained in a dual manner. 
Given a symmetry group, the basic idea 
 is similar as for SAT~(see also the French thesis of Jabbour~\cite{jabbour2008satisfiabilite} for a detailed discussion on lifting SAT symmetry breaking 
techniques to QBF).
First an order on $\strat_\exists(P)$ is imposed
such that every orbit contains an element which is minimal with respect to the order.
Then we construct a formula $\psi_\exists$
for which (at least) the minimal elements of the
orbits are winning strategies. Any such formula is an existential
symmetry breaker. One way of constructing an
existential symmetry breaker is given in the following theorem,
which generalizes the symmetry
breaking technique by Crawford et al.~\cite{DBLP:conf/kr/CrawfordGLR96}.
We give a formal proof that we obtain indeed a QBF symmetry breaker and 
conclude with lifting a recent CNF encoding to QBF.
\begin{thm}\label{thm:main}
  Let $P = Q_1x_1\ldots Q_nx_n$ be a prefix for~$X$, let $G\syn$ be a group acting admissibly on~$\BF(X)$,
  and let $G\sem$ be the associated group of~$G\syn$.
  Then
  \[
   \psi=\bigwedge_{\vbox{\clap{\scriptsize\strut$i=1$}\kern-3pt\clap{\scriptsize\strut$Q_i=\exists$}}}^n\quad\bigwedge_{g\in G\syn} \biggl(\Bigl(\bigwedge_{j<i} (x_j\leftrightarrow g(x_j))\Bigr)\to\Bigl(x_i\to g(x_i)\Bigr)\biggr)
  \]
  is an existential symmetry breaker for $G\syn$ and~$G\sem$.
\end{thm}
\begin{proof}
  All elements of $\strat_\exists(P)$ are trees with the same shape.
  Fix a numbering of the edge positions in these trees which is such that whenever
  two edges are connected by a path, the edge closer to the root has the smaller index.
  (One possibility is breadth first search order.)
  For any two distinct strategies $s_1,s_2\in\strat_\exists(P)$, there is then a minimal
  $k$ such that the labels of the $k$th edges of $s_1,s_2$ differ.
  Define $s_1<s_2$ if the label is $\bot$ for $s_1$ and $\top$ for $s_2$, and
  $s_1>s_2$ otherwise.

  Let $s\in\strat_\exists(P)$. We need to show that there are $g\syn\in G\syn$ and
  $g\sem\in G\sem$ such that $[g\syn(\psi)]_{g\sem(s)}=\top$.
  Let $g\syn=\id$ and let $g\sem$ be such that $\tilde s:=g\sem(s)$ is as small as possible
  in the order defined above. We show that $[\psi]_{\tilde s}=\top$.
  Assume otherwise. Then there exists $i\in\{1,\dots,n\}$ with $Q_i=\exists$ and $g\in G\syn$
  and a path $\sigma$ in $\tilde s$ with $[x_j]_\sigma=[g(x_j)]_\sigma$ for all $j<i$ and
  $[x_i]_\sigma=\top$ and $[g(x_i)]_\sigma=\bot$.
  By Lemma~\ref{lem:p}, the element $g\in G\syn$ can be translated into an element $f\in G\sem$ which maps
  $\tilde s$ to a strategy $f(\tilde s)$ which contains a path that agrees with $\sigma$
  on the upper $i-1$ edges but not on the $i$th.
  By Lemma~\ref{lem:10}, applied to the subgroup $H\subseteq G\syn$ consisting of
  all $h\in G\syn$ with $[x_j]_\sigma=[h(x_j)]_\sigma$ for all $j<i$, we may assume that
  $f(\tilde s)$ and $\tilde s$ only differ in edges that belong to the subtree rooted
  at the $i$th node of~$\sigma$. As all these edges have higher indices, we have
  $\tilde s < s$, in contradiction to the minimality assumption on~$s$.
  \qed
\end{proof}
Note that we do not need to know the group $G\sem$ explicitly.
It is only used implicitly in the proof.
In nontrivial applications, $G\syn$ will have a lot of elements. It is not necessary (and
not advisable) to use them all, although Thm.~\ref{thm:main} would allow us to do so. In
general, if a formula $\psi_1\land\psi_2$ is an existential symmetry breaker, then so are
$\psi_1$ and $\psi_2$, so we are free to use only parts of the large conjunctions. A
reasonable choice is to pick a set $E$ of generators for $G\syn$ and let the inner
conjunction run over (some of) the elements of $E$.

The formula~$\psi$ of Thm.~\ref{thm:main} can be efficiently encoded as conjunctive
normal form~(CNF), adopting the propositional encoding of \cite{DBLP:conf/sat/Devriendt0BD16,DBLP:series/faia/Sakallah09}:
let $g\in G\syn$ and let $\{y_0^g,\ldots,y_{n-1}^g\}$ be a set of fresh variables. First, we define a
set $I^g$ of clauses that represent all implications $x_i \rightarrow g(x_i)$
of $\psi$ from Thm.~\ref{thm:mainex},
\[
 I^g = \{(\neg y_{i-1}^g \lor \neg x_i \lor g(x_i)) \mid 1\leq i \leq n, Q_i = \exists\}.
\]
When $x_i$ is existentially quantified, by using Tseitin variables 
$y_{i-1}^g$ we can recycle the
implications $x_i\to g(x_i)$ in the encoding
of the equivalences $x_j \leftrightarrow g(x_j)$ that appear in the outer implication:
\[
E^g = \{(y_j^g \lor \neg y_{j-1}^g \lor \neg x_j)\land (y_j \lor \neg y_{j-1} \lor g(x_j)) \mid 1\leq j < n,Q_j = \exists\}.
\]
If variable $x_j$ is universally quantified, the recycling is not possible,
so we use
\[U^g = \{(y_j^g \lor \neg y_{j-1}^g \lor \neg x_j \lor \neg g(x_j))\land (y_j^g \lor  \neg y_{j-1}^g \lor x_j \lor g(x_j))  \mid 1\leq j < n,Q_j = \forall\}
\]
instead. The CNF encoding of $\psi$ is then the conjunction of $y_0^g$ and all the clauses
in $I^g$, $E^g$, and $U^g$, for all desired $g\in G\syn$.
The prefix $P$ has to be extended by additional quantifiers which bind the Tseitin variables~$y_i^g$.
As explained in~\cite{DBLP:conf/sat/EglySTWZ03},
the position of such a new variable in the prefix has to be behind the 
quantifiers of the variables occuring in its definition. 
The encoding of universal symmetry breakers works similarly and results in
a formula in disjunctive normal form (DNF), i.e., a
disjunction of cubes (conjunctions of literals). 
In this case the auxiliary variables are universally quantified.
The obtained cubes could be
used by solvers that simulatanously reason on the CNF and DNF
representation of a formula (e.g.,~\cite{DBLP:conf/date/GoultiaevaSB13,DBLP:journals/ai/JanotaKMC16}) or by solvers that
operate on formulas of arbitrary structure (e.g.,~\cite{DBLP:journals/corr/abs-1710-02198,DBLP:conf/sat/Tentrup16,DBLP:journals/ai/JanotaKMC16}).
The practical evaluation of this approach is a separate topic which we leave to future work.

Besides the practical evaluation of the discussed 
symmetry breakers in 
connection with recent QBF solving technologies there are many
more promising directions for future work.
Also different orderings than the lexicographic order applied
in Thm.~\ref{thm:main} could be used~\cite{DBLP:conf/cp/NarodytskaW13}
for the construction of novel symmetry breakers. 
Recent improvements of static symmetry 
breaking~\cite{DBLP:conf/sat/Devriendt0BD16} for 
SAT could be lifted to QBF and applied in combination with 
recent preprocessing techniques. 
Also dynamic symmetry 
breaking during the solving could be beneficial, for example 
in the form of symmetric explanation learning~\cite{DBLP:conf/sat/Devriendt0B17}.

An other interesting direction would be the relaxation of the 
quantifier ordering. Our symmetry framework 
assumes a fixed quantifier prefix with a strict ordering. In recent
works it has been shown that relaxing this order by the means 
of dependency schemes is beneficial for QBF solving both in theory and in 
practice~\cite{DBLP:conf/sat/BlinkhornB17,DBLP:conf/sat/PeitlSS17}. 
In a similar way as proof systems have been parameterized 
with dependency schemes, our symmetry framework 
can also be parameterized with dependency schemes. 
It can be expected that a more relaxed
notion of quantifier dependencies induces more symmetries
resulting in more powerful symmetry breakers.

\bibliographystyle{splncs}
\bibliography{refs}

\end{document}